\documentclass{fundam}

\usepackage{url}
\usepackage{graphicx}
\usepackage{tikz}
\usepackage[caption=false]{subfig}
\usepackage{amssymb}
\usepackage[colorlinks = true, linkcolor = black, urlcolor = black, citecolor = black, anchorcolor = black]{hyperref}
\usepackage{algorithm}
\usepackage{algorithmic}
\usepackage{booktabs}
\usepackage{array}
\usepackage[section]{placeins} 
\usetikzlibrary{decorations.pathreplacing,calc}
\newtheorem{problem}{Problem}

\newcommand{\pd}{\delta}                 
\newcommand{\mwm}{\mathit{mwm}}          

\begin{document}

\title{Exact Distances for Parameterized String Matching with Mismatches}
\author{Apurba Saha\thanks{These authors contributed equally to this work}, Iftekhar Hakim Kaowsar\footnotemark[1], Mahdi Hasnat Siyam\footnotemark[1], M. Sohel Rahman\footnotemark[1] \\ Department of Computer Science and Engineering \\ Bangladesh University of Engineering and Technology (BUET)
}
\maketitle

\runninghead{Saha, Kaowsar, Siyam, Rahman}{Parameterized Matching with Mismatches}
\begin{abstract}
  Two strings are considered to have parameterized matching when there exists a bijection of the parameterized alphabet onto itself such that it transforms one string to another. Parameterized matching has application in software duplication detection, image processing, and computational
biology. We consider the problem for which a pattern $p$, a text $t$ and a mismatch tolerance limit $k$ are given and the goal is to find all positions in text $t$ for which pattern $p$ parameterized matches with $|p|$-length substrings of $t$ with at most $k$ mismatches. The number of mismatches is the Hamming distance under the best bijection, which we call the parameterized distance. Our main result is an algorithm that computes the parameterized distance at every position of the text, and hence solves this problem for every $k$ at once. Its running time is $O((\sigma_p^2+\sigma_s)\, n\log m + n\,T_M)$, where $n = |t|$, $m=|p|$, $\sigma_p$ and $\sigma_s$ are the numbers of parameterized and static characters, and $T_M$ is the time to compute one maximum weight matching in a bipartite graph with $2\sigma_p$ vertices. With the algorithm of Gabow and Tarjan, $T_M=O(\sigma_p^{2.5}\log(\sigma_p m))$. Since this bound does not depend on $k$, our algorithm improves the algorithm by Hazay, Lewenstein and Sokol for $k=\tilde{\Omega}(\sigma_p^{5/3}+\sigma_s^{2/3})$.
\end{abstract}


\section{Introduction}
In parameterized string matching, a string is assumed to contain both `fixed' and `parameterized' characters. So, the underlying alphabet $\Sigma$ consists of two kinds of characters: static characters ($s$-characters; belonging to the set $\Sigma_s$) and parameterized characters ($p$-characters; belonging to the set $\Sigma_p$). In what follows, unless otherwise specified, we will assume the strings in this setting. Given two strings, denoted as $s$ and $s'$, a \textit{parameterized match} or $p$-match for short between them exists, if there is a bijection between the alphabets thereof. For example, let $\Sigma_s=\{\mathtt{X},\mathtt{Y}\}$ and $\Sigma_p=\{\mathtt{a},\mathtt{b},\mathtt{c}\}$. The strings $\mathtt{abXbca}$ and $\mathtt{baXacb}$ $p$-match, as the bijection that swaps $\mathtt{a}$ and $\mathtt{b}$ (and maps $\mathtt{c}$ to itself) transforms the first one to the second one. On the other hand, $\mathtt{abXbca}$ and $\mathtt{aaXbca}$ do not $p$-match, since both $\mathtt{a}$ and $\mathtt{b}$ would have to be mapped to $\mathtt{a}$; neither do $\mathtt{abXbca}$ and $\mathtt{abYbca}$, since a static character can only match itself.

Extending this concept, we say that two strings with equal length $p$-match with mismatch tolerance $k$ if there is a $p$-match after discarding at most $k$ indices. For example, $\mathtt{abXbca}$ and $\mathtt{baXacc}$ do not $p$-match, but they $p$-match with mismatch tolerance $1$: after discarding the last position, the swap of $\mathtt{a}$ and $\mathtt{b}$ works again. The minimum number of positions that have to be discarded is the \textit{parameterized distance} of the two strings. Equivalently, it is the smallest Hamming distance between the two strings after one of them is renamed by such a bijection; we define it formally in Section~\ref{sec:preliminaries}. Now we can easily extend this idea to pattern matching with mismatches as follows. Given a text $t$ and a pattern $p$, we need to find all locations of $t$ where a $p$-match with pattern $p$ exists, having mismatch tolerance $k$. We call this problem the decision variant, since it only asks whether the parameterized distance is at most $k$ at each position.
\par
Apart from its inherent combinatorial beauty, the parameterized matching problem has some interesting applications as well. For example, it is helpful to detect duplicate code in large software systems. It has been motivated by the fact that programmers prefer to duplicate code in large software systems. In order
to avoid new bugs and revisions, they prefer to simply copy a working section of code written
by someone else, while it is encouraged to understand the principles of the working section of code \cite{mendivelso2020brief, fredriksson2006efficient}. It also has applications in computational biology and image processing. In image processing, colors can be mapped in the presence of errors \cite{hazay2007approximate}. Furthermore, parameterized matching has been used in solving graph isomorphism \cite{mendivelso2013solving}.
\par
Allowing mismatches is natural in these applications. When a programmer copies a piece of code and renames its identifiers, the copy $p$-matches the original if every identifier has been renamed consistently. If a single occurrence of an identifier is forgotten, the copy does not $p$-match the original, so exact parameterized matching cannot distinguish it from unrelated code. The number of mismatches, on the other hand, is $1$ for this copy and large for unrelated code. Such forgotten renamings are a known source of bugs; Li et al.~\cite{li2004cpminer} found many of them in the Linux and FreeBSD kernels.

\paragraph{Related work.}
Baker \cite{baker1993program} first introduced the idea of parameterized pattern matching to detect source code duplication in software. The problem of finding every occurrence of a parameterized string over a text is solved by Baker using the $p$-suffix tree \cite{baker1996parameterized}, which can be constructed in $O \big(|t| (\left| \Sigma_p \right|+ \log (|\Sigma_s| + |\Sigma_p|) ) \big)$ time.
Subsequently, Cole and Hariharan \cite{cole2000faster} improved the construction time of the $p$-suffix tree to $O(|t|)$. Apostolico et al.~\cite{apostolico2007parameterized} solved the parameterized string matching allowing mismatches problem when both $t$ and $p$ are run-length encoded and the pattern or the text contains only two distinct characters. Their algorithm runs in $O(|t| + ( r_t \times r_p) \alpha(r_t) \log(r_t) )$ time, where $r_t$ and $r_p$ denote the number of runs in the encodings for $t$ and $p$ respectively and $\alpha$ is the inverse of Ackermann's function. Apostolico, Erd\H{o}s and J\"uttner~\cite{apostolico2012parameterized} later extended this approach to alphabets of any size, in $O(r_t r_p\,\sigma_p^5\log^3\sigma_p)$ time, that is, $O(r_t r_p)$ for alphabets of constant size. In the introduction of the same paper, they remark that the parameterized distance at every position of the text can be computed with one maximum weight matching per position, in $O(nm(\sqrt{m}+\log n))$ time, and that for alphabets of constant size it is enough to try every bijection with convolutions, in $O(n\log n)$ time.

Hazay et al. solved this decision variant in $O(|t|k^{1.5} + |p|k \log (|p|))$ time \cite{hazay2007approximate}. They also showed that, for two strings of equal length, the problem is equivalent to finding a maximum weight matching in a bipartite graph. Their solution performs slower for large values of the tolerance limit $k$. Bourhis et al.~\cite{bourhis2023distances} showed that, as soon as insertions or deletions are allowed, already deciding whether two strings $p$-match with at most $k$ edit operations is NP-complete. The model we study allows only mismatches, and in this model deciding whether two strings $p$-match with at most $k$ mismatches takes polynomial time~\cite{hazay2007approximate}.

\paragraph{Our contributions.}
In this paper we revisit the parameterized matching problem. Our solution is deterministic and runs in $O((\sigma_p^2+\sigma_s)\, n\log m + n\,\sigma_p^{2.5}\log(\sigma_p m))$ time, where $n=|t|$, $m=|p|$, $\sigma_p=|\Sigma_p|$ and $\sigma_s=|\Sigma_s|$ (Section~\ref{sec:general}, Theorem~\ref{thm:main}). It improves the algorithm by Hazay et al. for $k=\tilde{\Omega}(\sigma_p^{5/3}+\sigma_s^{2/3})$. Note that our solution does not depend on $k$ and it can be easily parallelized. It computes, for every position of the text, the exact parameterized distance between the pattern and the substring of the text starting there, so the same computation answers the decision variant for every $k$ at once. The idea is simple: the best bijection for a window depends only on how often each character of the window is aligned with each character of the pattern. For each pair of characters, a single polynomial multiplication counts how often the two are aligned in every window, and we then find the best bijection of each window with a maximum weight matching, following Hazay et al.

Table~\ref{tab:results} compares our result with the previous ones. We also give simple lower and upper bounds on the weight of a maximum matching (Lemma~\ref{lem:bounds}), which can decide a window of the decision variant without computing a matching.

\paragraph{Organization.}
Section~\ref{sec:preliminaries} gives the definitions and states the problems. Section~\ref{sec:general} presents the algorithm, its analysis, the bounds that can avoid a matching in the decision variant, and a parallel version. Section~\ref{sec:conclusions} concludes.

\begin{table}[t]
\centering
\caption{Algorithms for parameterized matching with mismatches. Here $n=|t|$, $m=|p|$, $r_t$ and $r_p$ are the numbers of runs of $t$ and $p$ when they are run-length encoded (RLE), and $T_M=O(\sigma_p^{2.5}\log(\sigma_p m))$ is the time of one maximum weight matching with the algorithm of Gabow and Tarjan~\cite{gabow1989faster}.}
\label{tab:results}
\footnotesize
\setlength{\tabcolsep}{4pt}
\begin{tabular}{@{}lll@{}}
\toprule
Algorithm & Time & Output \\
\midrule
Apostolico, Erd\H{o}s, Lewenstein~\cite{apostolico2007parameterized} & $O(n + r_t r_p\,\alpha(r_t)\log r_t)$ & exact distances, binary pattern or text \\
Apostolico, Erd\H{o}s, J\"uttner~\cite{apostolico2012parameterized} & $O(r_t r_p\,\sigma_p^5\log^3\sigma_p)$ & decision for one $k$ \\
Hazay, Lewenstein, Sokol~\cite{hazay2007approximate} & $O(nk^{1.5} + mk\log m)$ & decision for one $k$ \\
Direct matching~\cite{apostolico2012parameterized} & $O(nm(\sqrt{m}+\log n))$ & exact distances \\
\textbf{This paper} & $O((\sigma_p^2+\sigma_s)\, n\log m + n\,T_M)$ & exact distances \\
\bottomrule
\end{tabular}
\end{table}

\section{Preliminaries}\label{Preliminaries}\label{sec:preliminaries}
We follow the usual notations from the stringology literature. A string of length $n$, $s = s_{1}s_{2}...s_{n}$, is a sequence of characters drawn from an alphabet $\Sigma$. We use $s[i..j]= s_{i}s_{i+1}...s_{j}$ to denote a substring of length $j-i+1$; if $i=1$ ($j=n$), then $s[i..j]$ is a prefix (suffix) of $s$. Throughout this manuscript, we use $t$ and $p$ to denote text and pattern strings, respectively, and we let $n=|t|$ and $m=|p|$, where $m\le n$. The alphabet is $\Sigma=\Sigma_s\cup\Sigma_p$ with $\Sigma_s\cap\Sigma_p=\emptyset$, and we write $\sigma_s=|\Sigma_s|$, $\sigma_p=|\Sigma_p|$ and $\sigma=\sigma_s+\sigma_p$.

\begin{definition}\label{def:pmatch}
Two strings $s, s' \in (\Sigma_s \cup \Sigma_p)^n$ are said to $p$-match if a one-to-one function $\pi:\Sigma_s \cup \Sigma_p \to \Sigma_s \cup \Sigma_p$ exists that is the identity on $\Sigma_s$ and transforms $s$ to $s'$.
\end{definition}

\begin{definition}\label{def:distance}
Let $x$ and $y$ be two strings of equal length $\ell$. The \emph{parameterized distance} of $x$ and $y$ is
\[
\pd(x,y)=\min_{\pi}\ \bigl|\{\, i\in[1,\ell] \;:\; \pi(x_i)\neq y_i \,\}\bigr|,
\]
where the minimum is taken over all bijections $\pi:\Sigma\to\Sigma$ that are the identity on $\Sigma_s$.
\end{definition}

In other words, $\pd(x,y)$ is the smallest Hamming distance between $y$ and a renamed copy $\pi(x)=\pi(x_1)\pi(x_2)\ldots\pi(x_\ell)$ of $x$; this is the model used by Apostolico et al.~\cite{apostolico2007parameterized} and Hazay et al.~\cite{hazay2007approximate}. For a fixed $\pi$, the positions with $\pi(x_i)\ne y_i$ are the ones that must be discarded, so $x$ and $y$ $p$-match with mismatch tolerance $k$ if and only if $\pd(x,y)\le k$.

We call the substring $t[i..i+m-1]$ the \emph{window} of $t$ at position $i$, and we say that $p$ is \emph{aligned} at position $i$ of $t$. We write $\pd_i=\pd(t[i..i+m-1],p)$ for $i\in[1,n-m+1]$.

We end this brief section with a formal definition of the problems we tackle in this paper.
\begin{problem}\label{problem:decision}
    Given a text $t$, a pattern $p$ and a tolerance value $k$, determine for every $i\in [1,n-m+1]$ whether there is a parameterized match between the substring $t_it_{i+1}$$\dots$ $t_{i+m-1}$ and the pattern $p$, having no more than $k$ mismatches, that is, whether $\pd_i\le k$.
\end{problem}

\begin{problem}\label{problem:distances}
Given a text $t$ and a pattern $p$, compute $\pd_i$ for every $i\in[1,n-m+1]$.
\end{problem}

Clearly, a solution of Problem~\ref{problem:distances} solves Problem~\ref{problem:decision} for every $k$ at once. Our algorithm solves Problem~\ref{problem:distances}.

\section{Parameterized Matching allowing Mismatches}\label{Matching with Any Number of Mismatch}\label{sec:general}
Our algorithm has two steps. For a single window, the parameterized distance is given by a maximum weight matching in a bipartite graph on the parameterized characters (Section~\ref{sec:onewindow}). The edge weights of these graphs can be computed for all windows at once with polynomial multiplication (Section~\ref{sec:weights}).

\subsection{Reduction to maximum weight matching}\label{sec:onewindow}
Suppose we are given two strings $x$ and $y$ of equal length $\ell$. Clearly, if we can choose a subset of positions from $[1,\ell]$ of size at most $k$ and discard those positions in both strings to get $x'$ and $y'$ respectively so that strings $x'$ and $y'$ have a $p$-match, then we are done. The difficulty is that both the bijection and the positions to discard depend on the window, as the following lemma shows.

\begin{lemma}\label{lem:windows}
    Let $p$ be aligned at two positions $i_1\ne i_2$ of the text. A bijection that minimizes the number of mismatches at position $i_1$ may not minimize it at position $i_2$, and the positions of the text that have to be discarded at $i_1$ may not have to be discarded at $i_2$. This lemma also applies to patterns.
\end{lemma}
\begin{proof}
We will prove this for text by presenting a counter-example.

Let $t=\mathtt{abcbbbaaaca}$, $p=\mathtt{deeeef}$ and $k=2$. We assume there are no static characters.
Figure \ref{fig:setting_at_1} shows the illustration when $p$ is set at position $i=1$ of the text. We see that we must discard the $3$rd and $6$th positions of the text (similarly, the $3$rd and $6$th positions of the pattern) to have a parameterized match, and the bijection must pair $\mathtt{e}$ with $\mathtt{b}$ and $\mathtt{d}$ with $\mathtt{a}$. This is the only possible way. Whereas, when $p$ is placed at position $i=6$ of the text, we must discard the 10th and 11th positions of the text (similarly, the 5th and 6th positions of the pattern), and the bijection must pair $\mathtt{e}$ with $\mathtt{a}$ and $\mathtt{d}$ with $\mathtt{b}$. Figure \ref{fig:setting_at_6} shows the illustration.
\end{proof}

This is different from the classic pattern matching with mismatches, where a position is a mismatch or not independently of all the other positions. In parameterized matching, every window needs its own bijection.

\begin{figure}[ht]
    \centering
    \subfloat[Aligning $p$ at $i=1$ of $t$]{
        \begin{tikzpicture}[scale=0.5, every node/.style={minimum height=1cm, text depth=0.5ex}]
          \foreach \i/\char in {1/a,2/b,3/c,4/b,5/b,6/b,7/a,8/a,9/a,10/c,11/a} {
            \draw (\i, 2) rectangle (\i+1, 3) node[midway] {\char};
          }
          \foreach \i/\char in {1/d,2/e,3/e,4/e,5/e,6/f} {
            \draw (\i, 0) rectangle (\i+1, 1) node[midway] {\char};
          }
          \foreach \i/\char in {1/\checkmark,2/\checkmark,3/$\times$,4/\checkmark,5/\checkmark,6/$\times$} {
            \draw (\i, -2) rectangle (\i+1, -1) node[midway] {\char};
          }
          \node[anchor=west] at (1,-2.9) {\small bijection: $\mathtt{d}\leftrightarrow\mathtt{a}$, $\mathtt{e}\leftrightarrow\mathtt{b}$};
        \end{tikzpicture}
        \label{fig:setting_at_1}
    }\hfill
    \subfloat[Aligning $p$ at $i=6$ of $t$]{
    \begin{tikzpicture}[scale=0.5, every node/.style={minimum height=1cm, text depth=0.5ex}]
      \foreach \i/\char in {1/a,2/b,3/c,4/b,5/b,6/b,7/a,8/a,9/a,10/c,11/a} {
        \draw (\i, 2) rectangle (\i+1, 3) node[midway] {\char};
      }
      \foreach \i/\char in {1/d,2/e,3/e,4/e,5/e,6/f} {
        \draw (\i+5, 0) rectangle (\i+6, 1) node[midway] {\char};
      }
      \foreach \i/\char in {1/\checkmark,2/\checkmark,3/\checkmark,4/\checkmark,5/$\times$,6/$\times$} {
        \draw (\i+5, -2) rectangle (\i+6, -1) node[midway] {\char};
      }
      \node[anchor=west] at (6,-2.9) {\small bijection: $\mathtt{d}\leftrightarrow\mathtt{b}$, $\mathtt{e}\leftrightarrow\mathtt{a}$};
    \end{tikzpicture}
    \label{fig:setting_at_6}
    }
    \caption{(a) Text $t=\mathtt{abcbbbaaaca}$ and pattern $p=\mathtt{deeeef}$ have been aligned at the 1st position. Here $k=2$ and the only possible way of matching has been shown. The 3rd and 6th positions have been discarded. (b) The same pattern has been aligned at the 6th position. We see that this time the 5th and 6th positions have been discarded to have a parameterized match.}
    \label{fig:alignment}
\end{figure}

We will be constructing a bipartite graph and utilizing the concept of maximum matching as can be seen from the following lemma, which is also given by Hazay et al.~\cite{hazay2007approximate}.
\begin{lemma}\label{lem:matching}
    Consider two strings $x$ and $y$ of equal length $\ell$. We construct a weighted bipartite graph $G=(U \cup V, E)$ where $U= V = \Sigma_p$; it is complete and undirected, and its vertices are labelled by the parameterized characters. The weight $w(a,b)$ of the edge between $a\in U$ and $b \in V$ is the number of positions $i$ such that $x_i=a$ and $y_i=b$. Let $h(x,y)$ be the number of positions $i$ such that $x_i=y_i\in\Sigma_s$. Then
    \[
    \pd(x,y)=\ell-h(x,y)-\mwm(G),
    \]
    where $\mwm(G)$ refers to the weight of a maximum weight matching in graph $G$.
\end{lemma}
\begin{proof}
Let $\pi$ be a bijection that is the identity on $\Sigma_s$. Then $\pi$ maps $\Sigma_p$ onto $\Sigma_p$, so $M_\pi=\{(a,\pi(a)) : a\in\Sigma_p\}$ is a perfect matching of $G$, and every perfect matching of $G$ is $M_\pi$ for exactly one such $\pi$. Under $\pi$, a position $i$ is not a mismatch if and only if either $x_i=y_i\in\Sigma_s$, or $x_i\in\Sigma_p$ and $\pi(x_i)=y_i$. The positions of the first kind are counted by $h(x,y)$ and the second kind by the edges of $M_\pi$, so the number of positions that are not mismatches under $\pi$ is
\begin{equation}\label{eq:nonmismatches}
h(x,y)+\sum_{a\in\Sigma_p}w(a,\pi(a))=h(x,y)+w(M_\pi).
\end{equation}
By Definition~\ref{def:distance}, $\ell-\pd(x,y)$ is the maximum of \eqref{eq:nonmismatches} over all such $\pi$, hence
\[
\ell-\pd(x,y)=h(x,y)+\max_{M\text{ perfect}}w(M)=h(x,y)+\mwm(G).
\]
Since $G$ is complete and its weights are non-negative, we can add edges to a maximum weight matching until it is perfect without decreasing its weight, so the last equality holds.
\end{proof}

\begin{figure}[ht]
\centering
\begin{tikzpicture}[thick, x=2.2cm, y=2.6cm]
  \tikzset{node style/.style={circle,draw,minimum size=2em},
           w/.style={fill=white, inner sep=1.5pt, font=\small}}
  \foreach \c [count=\j from 0] in {a,b,c,d,e} {
    \node[node style] (\c) at (\j,1) {\c};
    \node[node style] (\c') at (\j,0) {\c};
  }
  \node[anchor=east] at (-0.35,1) {$U$ (characters of $x$)};
  \node[anchor=east] at (-0.35,0) {$V$ (characters of $y$)};

  \draw (a) -- node[w,pos=0.5]{1} (a');
  \draw[line width=2pt] (a) -- node[w,pos=0.18]{2} (e');
  \draw[line width=2pt] (b) -- node[w,pos=0.12]{3} (d');
  \draw (c) -- node[w,pos=0.55]{1} (a');
  \draw[line width=2pt] (c) -- node[w,pos=0.5]{1} (b');
  \draw[line width=2pt] (d) -- node[w,pos=0.5]{1} (c');
  \draw[line width=2pt] (e) -- node[w,pos=0.18]{2} (a');
\end{tikzpicture}
\caption{Example of constructing the matching graph with $x=\mathtt{abcaaeebbcd}$ and $y=\mathtt{adbeeaaddac}$. The weight of the edge $(u,v)$ is the number of positions where $u$ and $v$ are aligned. Edges of weight $0$ are not shown. Matched edges are thickened.}
\label{fig:example_string_graph_1}
\end{figure}
\begin{example}\label{ex:graph}
Let us see an example of constructing the graph for two strings of equal length $x=\mathtt{abcaaeebbcd}$ and $y=\mathtt{adbeeaaddac}$. Figure \ref{fig:example_string_graph_1} shows the graph. Edges of the maximum matching have been thickened.

In this graph, the weight of the edge $(u, v)$ is the number of positions $i$ such that $x_i=u$ and $y_i=v$. For example, at positions $i \in \{2,8,9\}$ we have $x_i=\mathtt{b}$ and $y_i=\mathtt{d}$. Hence, the edge $(\mathtt{b},\mathtt{d})$ has a weight of $3$. A maximum matching in this bipartite graph chooses the edges $(\mathtt{a},\mathtt{e})$, $(\mathtt{b},\mathtt{d})$, $(\mathtt{c},\mathtt{b})$, $(\mathtt{d},\mathtt{c})$, $(\mathtt{e},\mathtt{a})$, with $\mwm(G)=9$. It means that only positions $\{2,3,4,5,6,7,8,9,11\}$ will be preserved and positions $\{1,10\}$ will be discarded. So, the minimum number of mismatches is $2$.
\end{example}

For solving our problem, we can construct the graph $G_i$ as described above for every position $i \in [1,n-m+1]$ of the text, with $x=t[i..i+m-1]$ and $y=p$. If for an $i$ we have $\pd_i=m-h_i-\mwm(G_i)\le k$, where $h_i=h(t[i..i+m-1],p)$ (Lemma~\ref{lem:matching}), we can report that there is a parameterized matching for $p$ in the corresponding window.

\paragraph{Complexity analysis.}
Once the weights are known, computing all the distances takes $n-m+1$ maximum weight matchings, one for each $G_i$, in $O(n\,T_M)$ time, where $T_M$ is the time of one matching. Each $G_i$ has $2\sigma_p$ vertices, and as its weights sum to at most $m$, at most $\min(\sigma_p^2,m)$ of its edges have a non-zero weight. The algorithm of Gabow and Tarjan~\cite{gabow1989faster} runs in $O(\sqrt{|V|}\,|E|\log(|V|W))$ time on a graph with $|V|$ vertices, $|E|$ edges and integer weights at most $W$, so
\[
T_M=O(\sqrt{\sigma_p}\,\min(\sigma_p^2,m)\log(\sigma_p m))
\]
and the matchings take $O(n\,\sigma_p^{2.5}\log(\sigma_p m))$ time in total. One matching needs $O(\sigma_p^2)$ memory, which is reused for all windows.

\subsection{Computing the edge weights}\label{sec:weights}
We are now left with finding the weights of the edges for every window of the text efficiently. We use polynomial multiplication for this purpose. We define $P\oplus Q$ as the product of two polynomials $P$ and $Q$; for example, $(1+z)\oplus(1+z^2)=1+z+z^2+z^3$. With the Fast Fourier Transform (FFT), it can be computed in $O(D\log D)$ time~\cite{brigham1988fast}, where $D$ is the degree of the product. For each $a \in \Sigma$, we create a polynomial $P(a)$ ($Q(a)$) of $z$ with degree $n-1$ ($m-1$) as follows:
\begin{equation}\label{Pa}
P(a)=c_0(a)+c_1(a)z+c_2(a)z^2+\dots+c_{n-1}(a)z^{n-1}
\end{equation}
\begin{equation}\label{Qa}
Q(a)=d_0(a)+d_1(a)z+d_2(a)z^2+\dots+d_{m-1}(a)z^{m-1}
\end{equation}
Here, $c_j(a)=[t_{j+1}=a]$ for $j\in[0,n-1]$ and $d_j(a)=[p_{m-j}=a]$ for $j\in[0,m-1]$, where $[A]$ is the Iverson bracket.
\begin{lemma}\label{lem:fft}
    Let $a,b\in \Sigma_p$ and $R(a,b)=P(a) \oplus Q(b)$. The coefficient of $z^{m+i-1}$ in $R(a,b)$ will be the weight of the edge $(a,b)$ in the graph $G_{i+1}$.
\end{lemma}
\begin{proof}
\begin{equation*}
\begin{split}
     &\text{ Coefficient of $z^{m+i-1}$ in $R(a,b) $} \\
     & = \displaystyle \sum_{j=0}^{n-1} \sum_{k=0}^{m-1} c_j(a) d_k(b) [j+k=m+i-1] \\
     &= \displaystyle \sum_{k=0}^{m-1} \sum_{j=0}^{n-1} c_j(a) d_k(b) [j+k=m+i-1]\\
     &= \displaystyle \sum_{k=0}^{m-1} c_{m+i-1-k}(a) d_k(b) \\
     &= \displaystyle \sum_{k=0}^{m-1} [t_{m+i-k}=a] [p_{m-k}=b]\\
     &= \displaystyle \sum_{k'=1}^{m} [t_{i+k'}=a] [p_{k'}=b]\\
     &= \displaystyle \sum_{k'=1}^{m} [t_{i+k'}=a ~\wedge~ p_{k'}=b]\\
     &= \text{weight of the edge $(a,b)$ in the graph $G_{i+1}$}
\end{split}
\end{equation*}
Hence, the result follows.
\end{proof}
Intuitively, as we place the coefficients of $P(a)$ walking forward (from left to right) in $t$ and the coefficients of $Q(b)$ walking backward (from right to left) in $p$, the coefficients of $R(a,b)$ capture the number of positions where $a$ faces $b$ for the same window. Similarly, if $a \in \Sigma_s$ (i.e., for static characters), the coefficient of $z^{m+i-1}$ in $R(a)=P(a) \oplus Q(a)$ will be the number of positions where the same static character $a$ is aligned in the $(i+1)$-th window.

\begin{example}\label{ex:polynomials}
Consider the same example of Figure \ref{fig:alignment}, where $t=\mathtt{abcbbbaaaca}$, $p=\mathtt{deeeef}$ and $\Sigma=\{\mathtt{a},\allowbreak\mathtt{b},\allowbreak\mathtt{c},\allowbreak\mathtt{d},\allowbreak\mathtt{e},\allowbreak\mathtt{f}\}$.
\begin{equation*}
\begin{split}
         P(\mathtt{a}) & = 1+z^6+z^7+z^8+z^{10}\\
         P(\mathtt{b}) & = z+z^3+z^4+z^5\\
         P(\mathtt{c}) & = z^2+z^9\\
         P(\mathtt{d}) & = P(\mathtt{e}) = P(\mathtt{f}) = 0 \\ \\
         Q(\mathtt{a}) & = Q(\mathtt{b}) = Q(\mathtt{c}) = 0 \\
         Q(\mathtt{d}) & = z^5 \\
         Q(\mathtt{e}) & = z+z^2+z^3+z^4\\
         Q(\mathtt{f}) & = 1
\end{split}
\end{equation*}
For example, $R(\mathtt{a},\mathtt{e})$ is computed as follows.
\begin{equation*}
\begin{split}
 R(\mathtt{a},\mathtt{e}) &= P(\mathtt{a}) \oplus Q(\mathtt{e})\\
        &=
        z+z^2+z^3+z^4+z^7+2z^8+3z^9+ \\
        &\hphantom{{}=} 3z^{10}+3z^{11}+2z^{12}+z^{13}+z^{14}\\
\end{split}
\end{equation*}
By Lemma~\ref{lem:fft}, the window at position $i+1$ corresponds to the coefficient of $z^{m+i-1}=z^{i+5}$. We can easily verify (from $R(\mathtt{a},\mathtt{e})$) that $\mathtt{a}$ and $\mathtt{e}$ are aligned $0,0,1,2,3$ and $3$ times for pattern $p$ aligned at positions $1,2,\dots,6$ respectively (Figure \ref{fig:alignment}).
\end{example}

So the weights of all windows are obtained from the products $R(a,b)$ for $a,b\in\Sigma_p$, and the values $h$ of Lemma~\ref{lem:matching} from the products $R(a,a)$ for $a\in\Sigma_s$, in total $\sigma_p^2+\sigma_s$ products. With FFT on polynomials of degree $n+m-2$, this takes $O((\sigma_p^2+\sigma_s)\,n\log n)$ time and $O((\sigma_p^2+\sigma_s)\,n)$ memory.

\paragraph{Optimizations.}
We can use two well-known optimizations to further improve this running time.

First, to reduce the time of each product from $O(n\log n)$ to $O(n\log m)$, we employ a standard trick described by Clifford and Clifford~\cite{clifford2007simple}. Let $L\ge 2m$ be an integer with $L=\Theta(m)$. The text is partitioned into blocks of $L$ consecutive characters, with the first block starting at the beginning of the text and each subsequent block overlapping the previous one in $m-1$ characters (Figure~\ref{fig:blocks}). By Lemma~\ref{lem:fft}, the coefficients for the windows at positions $i+1,\ldots,i+L-m+1$ depend only on the block $t[i+1..i+L]$. Hence, multiplying $Q(b)$ with the polynomial of every block gives the weights of all windows. The product for one block takes $O(m\log m)$ time, giving $O((n/m)\,m\log m)=O(n\log m)$ time over all $O(n/m)$ blocks.

Second, the FFT multiplies two polynomials by transforming both into point-value form, multiplying the values, and transforming the result back. We transform each $Q(b)$ and the polynomial of each character in a block only once, not once for every product. This only improves the constant factor; the asymptotic running time stays the same.

\begin{figure}[ht]
\centering
\begin{tikzpicture}[x=0.8cm,y=0.55cm]
  \draw[thick] (0,0) rectangle (14,0.8);
  \node[left] at (0,0.4) {$t$};
  \foreach \s/\y/\lbl in {0/1.5/1, 4/3.0/2, 8/4.5/3} {
    \fill[black!22] (\s,\y) rectangle (\s+4,\y+0.8);
    \draw (\s,\y) rectangle (\s+6,\y+0.8);
    \node at (\s+3,\y+0.4) {\small block \lbl};
  }
  \draw[densely dotted] (0,1.5) -- (0,0.8);
  \draw[densely dotted] (4,3.0) -- (4,0.8);
  \draw[densely dotted] (6,1.5) -- (6,0.8);
  \draw[decorate,decoration={brace,mirror,amplitude=4pt}] (4,-0.1) -- (6,-0.1) node[midway,below=4pt] {\small $m-1$};
  \draw[decorate,decoration={brace,amplitude=4pt}] (8,5.45) -- (14,5.45) node[midway,above=4pt] {\small $L$};
  \draw[decorate,decoration={brace,amplitude=4pt}] (0,2.4) -- (4,2.4) node[midway,above=4pt] {\small $L-m+1$};
\end{tikzpicture}
\caption{The text is cut into blocks of length $L$ that overlap in $m-1$ characters. The windows that start in the shaded part of a block ($L-m+1$ of them) lie inside that block, so multiplying $Q(b)$ with the polynomial of every block gives the weights of all windows.}
\label{fig:blocks}
\end{figure}

\paragraph{Complexity analysis.}
With the optimizations above, the weights and the values $h_i$ of all windows are computed in $O((\sigma_p^2+\sigma_s)\,n\log m)$ time. If the windows of a block are processed before the next block, we only keep the transforms of the polynomials $Q(b)$ and of the current block, so the memory is $O((\sigma_p^2+\sigma_s)\,m+n)$.

\subsection{Algorithm and complexity analysis}
Algorithm \ref{alg:general} shows the steps to compute $\pd_i$ for every window and to find all window positions with parameterized mismatch not greater than the tolerance value $k$. We are using FFT for polynomial multiplication in Lines \ref{line:ab} and \ref{line:aa}, block by block as described above.

\begin{algorithm}[!t]
\caption{Algorithm to Solve Parameterized Matching with Any Number of Mismatches}
\label{alg:general}
\begin{algorithmic}[1]
\REQUIRE Text $t=t_1t_2\cdots t_n$, pattern $p=p_1p_2\cdots p_m$, tolerance $k$.
\ENSURE $\pd_i$ for every $i\in[1,n-m+1]$, and all positions $i$ with $\pd_i\le k$.
\FOR { $a \in \Sigma$}
\STATE Construct the polynomials $P(a)$ and $Q(a)$ from $t$ and $p$ respectively. \COMMENT{$P(a)$ block by block}
\ENDFOR
\FOR{$a \in \Sigma_p $}
    \FOR{$b \in \Sigma_p$}
        \STATE $R(a,b):=P(a) \oplus Q(b)$ \COMMENT{FFT, block by block} \label{line:ab}
    \ENDFOR
\ENDFOR
\FOR{$a \in \Sigma_s $}
    \STATE $R(a,a):=P(a) \oplus Q(a)$ \label{line:aa}
\ENDFOR
\STATE $output:=\emptyset$
\FOR {$i=1,\dots,n-m+1$}
    \STATE $U:=\Sigma_p$, $V:=\Sigma_p$, $E:=\emptyset$
    \FOR{$a \in \Sigma_p $}
        \FOR{$b \in \Sigma_p$}
            \STATE $w:=$ coefficient of $z^{m+i-2}$ in $R(a,b)$ \COMMENT{weight of $(a,b)$ in $G_i$}
            \STATE $E:= E \cup \{(a,b,w)\}$
        \ENDFOR
    \ENDFOR
    \STATE $h_i:=$ sum of the coefficients of $z^{m+i-2}$ in $R(a,a)$ for $a\in\Sigma_s$ \COMMENT{static}
    \STATE $\pd_i:=m-h_i-\mwm(U\cup V,E)$ \COMMENT{Lemma~\ref{lem:matching}} \label{line:mat}
    \IF{$\pd_i\le k$}
        \STATE $output:=output \cup \{i\}$
    \ENDIF
\ENDFOR
\RETURN $(\pd_1,\ldots,\pd_{n-m+1})$ and $output$
\end{algorithmic}
\end{algorithm}

\begin{theorem}\label{thm:main}
Algorithm~\ref{alg:general} solves Problem~\ref{problem:distances} (and Problem~\ref{problem:decision}) in $O((\sigma_p^2+\sigma_s)\,n\log m + n\,T_M)$ time, where $T_M$ is the time to compute a maximum weight matching in a bipartite graph with $2\sigma_p$ vertices and at most $\min(\sigma_p^2,m)$ edges, whose weights are positive integers that sum to at most $m$. With the algorithm of Gabow and Tarjan~\cite{gabow1989faster}, $T_M=O(\sqrt{\sigma_p}\,\min(\sigma_p^2,m)\log(\sigma_p m))$, so the total time is $O((\sigma_p^2+\sigma_s)\,n\log m + n\,\sigma_p^{2.5}\log(\sigma_p m))$. If the windows of each block are processed right after the products of that block are computed, the algorithm uses $O((\sigma_p^2+\sigma_s)\,m+n)$ memory.
\end{theorem}
\begin{proof}
The correctness of the algorithm follows from Lemmas~\ref{lem:matching} and~\ref{lem:fft}. The products take $O((\sigma_p^2+\sigma_s)\,n\log m)$ time and $O((\sigma_p^2+\sigma_s)\,m+n)$ memory (Section~\ref{sec:weights}), and the matchings take $O(n\,T_M)$ time and $O(\sigma_p^2)$ memory (Section~\ref{sec:onewindow}).
\end{proof}

We remark that the recent almost-linear time algorithms for minimum cost flow~\cite{chen2022maximum,vandenbrand2023deterministic} can compute a maximum weight matching of such a graph in $\min(\sigma_p^2,m)^{1+o(1)}\log m$ time. As our graphs have at most $2\sigma_p$ vertices, this saves at most a factor of roughly $\sqrt{\sigma_p}$ over the algorithm of Gabow and Tarjan, which we use in Theorem~\ref{thm:main}.

\subsection{Avoiding the matchings in the decision variant}
If we only need to know whether $\pd_i\le k$ (Problem~\ref{problem:decision}), the following simple bounds can answer the question without computing a matching.
\begin{lemma}\label{lem:bounds}
Let $G$ be a bipartite graph with vertex sets $U$ and $V$ and a weight $w(a,b)\ge 0$ for every $a\in U$ and $b\in V$, and let $w(G)=\sum_{a\in U,b\in V}w(a,b)$ be its total weight. Let $\rho_U=\sum_{a\in U}\max_{b\in V}w(a,b)$ and $\rho_V=\sum_{b\in V}\max_{a\in U}w(a,b)$. Then
\[
\rho_U+\rho_V-w(G)\ \le\ \mwm(G)\ \le\ \min(\rho_U,\rho_V).
\]
\end{lemma}
\begin{proof}
A matching has at most one edge at every vertex $a\in U$, and this edge weighs at most $\max_{b\in V}w(a,b)$. Summing over $U$ gives $\mwm(G)\le\rho_U$, and summing over $V$ gives $\mwm(G)\le\rho_V$. This proves the upper bound $\mwm(G)\le\min(\rho_U,\rho_V)$.

For the lower bound $\mwm(G)\ge\rho_U+\rho_V-w(G)$, mark at every vertex one of its heaviest edges. The edges marked at their endpoint in $U$ weigh $\rho_U$ in total, so the edges not marked there weigh $w(G)-\rho_U$. In the same way, the edges not marked at their endpoint in $V$ weigh $w(G)-\rho_V$. Let $M$ be the set of edges marked at both endpoints. Since only one edge is marked at every vertex, $M$ is a matching. Every edge outside $M$ is not marked at its endpoint in $U$ or not marked at its endpoint in $V$, so the edges outside $M$ weigh at most $(w(G)-\rho_U)+(w(G)-\rho_V)$. Hence
\[
\mwm(G)\ \ge\ w(M)\ \ge\ w(G)-(w(G)-\rho_U)-(w(G)-\rho_V)\ =\ \rho_U+\rho_V-w(G).
\]
\end{proof}
Applying Lemma~\ref{lem:bounds} to $G_i$, whose total weight $w(G_i)$ is the number of positions $q$ where both $t_{i+q-1}$ and $p_q$ are parameterized characters, gives
\[
m-h_i-\min(\rho_U,\rho_V)\ \le\ \pd_i\ \le\ m-h_i-(\rho_U+\rho_V-w(G_i)),
\]
where $\rho_U$ and $\rho_V$ are computed for $G_i$.
The bounds take $O(\sigma_p^2)$ time per window. If the lower bound is larger than $k$ we reject the window, and if the upper bound is at most $k$ we report it, in both cases without computing a matching.

\subsection{Achieving faster runtime with parallelization}\label{Achieving faster run-time with parallelization}\label{sec:parallelization}
A significant advantage of our solution is that it can be parallelized to achieve a faster run-time for a fixed text and pattern. In lines \ref{line:ab} and \ref{line:aa}, we are multiplying two polynomials generated from a pair of characters independently. For every distinct pair and every block of Section~\ref{sec:weights}, this multiplication can be done in parallel. Similarly, in line \ref{line:mat}, we are finding the maximum bipartite weighted matching for all windows $i\in[1,n-m+1]$ independently. All of these matchings can be run in parallel. So, if the number of processors $C$ increases up to $\min(n-m+1, (\sigma_p^2+\sigma_s)B)$, where $B=\Theta(n/m)$ is the number of blocks (the first optimization in Section~\ref{sec:weights}), our runtime gets divided by $C$.

\subsection{Discussion}\label{sec:discussion}
Recall that Apostolico et al.~\cite{apostolico2007parameterized} solved this problem for run-length encoded strings where the pattern or the text contains only two distinct characters. In the worst cases, when $r_t=\Theta(n)$ and $r_p=\Theta(m)$ (for example, when no two consecutive characters are equal), their solution runs in $O(n + nm\,\alpha(n) \log n)$ time, where $\alpha$ is the inverse of Ackermann's function. The extension of Apostolico, Erd\H{o}s and J\"uttner~\cite{apostolico2012parameterized} to alphabets of any size runs in $O(r_t r_p\,\sigma_p^5\log^3\sigma_p)$ time, which is $O(nm)$ in the same worst cases if $\sigma_p$ is a constant. For alphabets of constant size, the time complexity of our solution is $O(n\log m)$. So our proposed solution performs better in this regard.

Additionally, the solution proposed by Hazay, Lewenstein and Sokol~\cite{hazay2007approximate} runs in $O(nk^{1.5} + mk \log m)$ time. Our algorithm improves it when $k=\tilde{\Omega}(\sigma_p^{5/3}+\sigma_s^{2/3})$, while their algorithm is faster for small $k$. Since both bounds can be evaluated beforehand from $n$, $m$, $k$, $\sigma_p$ and $\sigma_s$, for the decision variant one can simply run the faster of the two.

\section{Conclusions}
\label{sec:conclusions}
In this paper, we have revisited the parameterized string matching problem that allows a fixed mismatch tolerance. Our algorithm computes the exact parameterized distance for every position of the text, so it answers the decision variant for every $k$ at once. Our running time depends on $\sigma$ and not on $k$. Faro et al.~\cite{faro2026bitwise} recently extended the approach of Section~\ref{sec:general} to indeterminate strings, where every position of the text is a set of characters.

Future work can be done on the correlation of the polynomial multiplications and the bipartite matching calculations, done independently in this paper. Furthermore, in this paper, we have not used the fact that the total sum of non-zero coefficients in the formed polynomials is in $O(n)$. This fact can be brought in to have further improvements. When $\Sigma_p=\emptyset$, the problem becomes computing the Hamming distances between $p$ and the windows of $t$, which takes $O(n\sqrt{m\log m})$ time for any alphabet~\cite{abrahamson1987generalized}, or $O(n\sqrt{m})$ time with a randomized algorithm~\cite{chan2023faster}. For run-length encoded strings, Apostolico, Erd\H{o}s and J\"uttner~\cite{apostolico2012parameterized} leave open whether the dependence on the alphabet size is necessary. It would be interesting to know the same for our problem: whether the parameterized distances can be computed in time independent of the alphabet size, as the Hamming distances can, or whether the bijection makes the problem harder.

\bibliographystyle{fundam}
\bibliography{citations}
\end{document}